\documentclass[12pt]{article}

\usepackage{amsmath,amsfonts,amsthm,amscd,amssymb,slashed}
\usepackage[pdftex]{graphicx}
\usepackage{algpseudocode}
\usepackage{tikz}
\usepackage{enumerate}
\usepackage[hidelinks]{hyperref}
\usepackage{mathtools}
\usetikzlibrary{shapes,arrows}
\usetikzlibrary{matrix,arrows}
\usepackage{epstopdf}
\usepackage{slashed}
\usepackage{xcolor,import}
\usepackage{transparent}
\usepackage[all]{xy}

\usepackage{geometry}
\makeatletter

\@addtoreset{equation}{section}
\def\section{\@startsection {section}{1}{\z@}{-2.5ex plus -1ex minus
 -.2ex}{1.3ex plus .2ex}{\large\bf}}
\def\subsection{\@startsection{subsection}{2}{\z@}{-2.25ex plus%
 -1ex minus -.2ex}{0.5ex plus .2ex}{\bf}}

\newcommand{\R}{\mathbb{R}}
\newcommand{\N}{\mathbb{N}}
\newcommand{\C}{\mathbb{C}}
\newcommand{\Z}{\mathbb{Z}}

\newtheorem{thm}{Theorem}[section]
\newtheorem{prop}[thm]{Proposition}
\newtheorem{cor}[thm]{Corollary} 
\newtheorem{dfn}[thm]{Definition}

\theoremstyle{definition}

\newcommand{\lf}[1]{\sigma^{ #1}}

\numberwithin{equation}{section}

\begin{document}
\parskip 6pt
\parindent 0pt

\vspace{.4cm}

\begin{center}
{\Large \bf Cartan Connections and Integrable Vortex Equations}

\baselineskip 18pt

\vspace{0.4 cm}

{\bf Calum Ross}\\
{Department of Mathematics, University College London, London WC1E 6BT, United Kingdom\\
Department of Physics and Research and Education Center for Natural Sciences, Keio
University, Hiyoshi 4-1-1, Yokohama, Kanagawa 223-8521, Japan\\
calum.ross@ucl.ac.uk\\
December 2021}
\end{center}
\abstract{We demonstrate that integrable abelian vortex equations on constant curvature Riemann surfaces can be reinterpreted as flat non-abelian Cartan connections. By lifting to three dimensional group manifolds we find higher dimensional analogues of vortices. These vortex configurations are also encoded in a Cartan connection. We give examples of different types of vortex that can be interpreted this way, and compare and contrast this Cartan representation of a vortex with the symmetric instanton representation.}

\section{Introduction}
In this paper we will finish the program started in \cite{RS1}, and furthered in \cite{RS2}, of relating all five of the integrable abelian vortex equations from \cite{Manton1} to the geometry of three dimensional Lie groups. In particular we demonstrate that an abelian vortex is equivalent to a flat non-abelian connection.

This relationship showcases the role that constant curvature geometries play in the construction of explicit vortex configurations. The most common way to see this relationship is to consider the Taubes equation satisfied by the modulus of the Higgs field, in the integrable cases it reduces to the Liouville equation. The Liouville equation is satisfied by the conformal factor for a metric with constant Gauss curvature. This leads to the interpretation of the modulus of the Higgs field as a conformal factor for a constant curvature metric. This metric is often called the Baptista metric \cite{Baptista}.

The standard Abelian-Higgs model is a two dimensional model of gauged vortices. The model consists of a complex scalar field $\phi$ called the Higgs field and a $U(1)$ gauge potential $a$. On a Riemann surface $M_{0}$ with the conformal factor $\Omega_{0}$,  the Abelian-Higgs model at critical coupling has the static energy functional \cite{MS}
\begin{equation}
E=\frac{1}{2}\int\left(\frac{B^{2}}{\Omega_{0}^{2}}+\frac{1}{\Omega_{0}}\overline{D_{i}\phi}D^{i}\phi+\frac{1}{4}\left(1-|\phi|^{2}\right)^{2} \right)d\text{Vol},
\label{AH model}
\end{equation}
with $B=f_{12}=\partial_{1}a_{2}-\partial_{2}a_{1}$.
This can be rewritten using a Bogomol'nyi argument to see that the energy is bounded below,
\begin{equation}
E\geq \pi |N|,
\end{equation}
with $N$ the winding number of the field. When $N>0$ the minimisers solve the first order, Bogomol'nyi, equations
\begin{align}
D_{\bar{z}}\phi&= \left(\partial_{\bar{z}}-ia_{\bar{z}}\right)\phi=0, \label{covariant holomorphic}\\ 
B&=\frac{\Omega_{0}}{2}\left(1-|\phi|^{2}\right),\label{eq:hyperbolic vortex part 2}
\end{align}
called the vortex equations. For $N<0$ the first equation, \eqref{covariant holomorphic}, becomes $D_{z}\phi=0$.

Decomposing the Higgs field as $\phi=e^{u+i\chi}$, and taking account of the singularities of $h$ at the zeros, $Z_{r}$ $r=1,\dots , N$ possibly repeated, of the Higgs field, the Bogomol'nyi equations can be converted into the Taubes equation
\begin{equation}
-\frac{4}{\Omega_{0}}\partial_{z}\partial_{\bar{z}}u=1- e^{2u}-\frac{4\pi}{\Omega_{0}} \sum_{r=1}^{N}\delta(z-Z_{r}). \label{Taubes eq on RS}
\end{equation}

A detailed study of this equation, for the case of the Abelian-Higgs model on the plane, is given in \cite{JT}. From a mathematical point of view this model is constructed from the data of a Riemann surface $M_{0}$, a connection $a$ and a smooth complex section $\phi$ of a complex line bundle over $M_{0}$. The pair $(\phi,a)$ is called a vortex.

Fig.~\ref{fig:unified picture summary figure} summarises the general story  followed in this paper, and demonstrates the relationship between vortex configurations on Lie groups in the top left and vortices on Riemann surfaces in the bottom left. It guides how we proceed in this paper starting with explaining the vortex story along the bottom line, before moving on to the three dimensional story  which give the details of the upper part of the figure. 

\begin{figure}[!htbp]
 \centering
\includegraphics[width=8truecm]{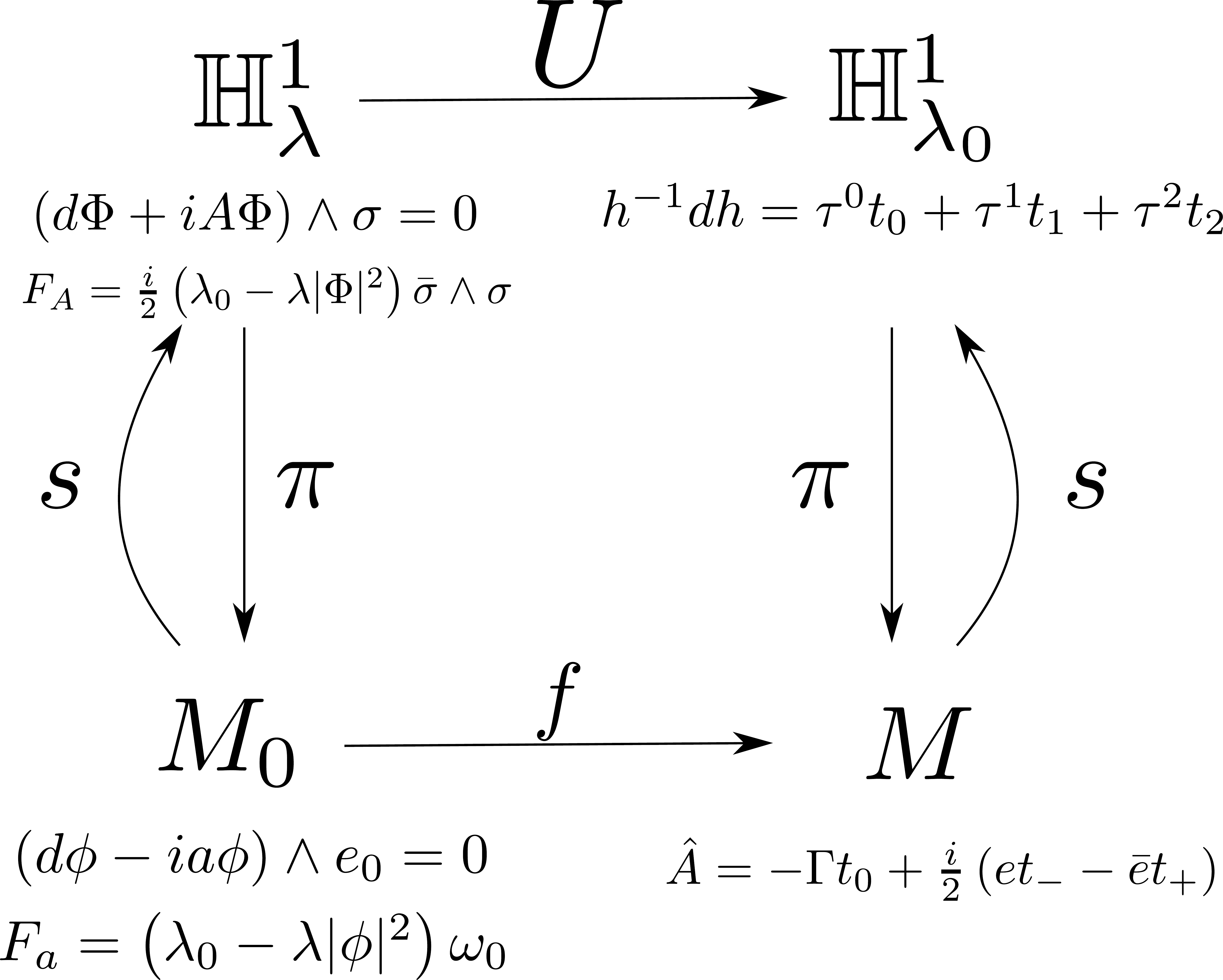}
 \caption{This summaries the four sets of equations and spaces that we relate in this paper. An extension on the left hand side relating the equations on $\mathbb{H}^{1}_{\lambda_{0}}$ to vortex magnetic modes on flat spaces can be included when $\lambda_{0}\neq 0$. This figure is adapted from one in \cite{Ross_thesis}.}
 \label{fig:unified picture summary figure}
\end{figure}

The paper is ordered as follows. In Sec.~\ref{sec:notation and conventions} we state our conventions for the  geometry of Lie groups and Riemann surfaces. Following this we demonstrate how to encode the structure and Gauss equations within a non-abelian flat connection, and how this flat connection descends from the Maurer-Cartan form on a Lie group.  This describes the right hand side of Fig.~\ref{fig:unified picture summary figure}.

Sec.~\ref{sec:vortex equations} summarises results about the integrable abelian vortex equations of \cite{Manton2}, as well as a discussion of the geometric interpretations of these vortices. These geometric interpretations are; as deformations of the metric to introduce degeneracies, an idea introduced in \cite{Baptista}, and as non-abelian Cartan connections encoding a degenerate frame.

Next in Sec.~\ref{sec:vortex configurations} we introduce and study the three dimensional generalisation of a vortex, a vortex configuration. Vortex configurations are also given by a flat connection, this time the pull back of the Maurer-Cartan one-form by a bundle map. This relationship between vortex configurations and flat connections is the key result of the section, and we use it to relate vortex configurations to vortices. Sec.~\ref{sec:magnetic modes} is a comment about the construction of solutions to massless Dirac equations from vortices. It includes extensions of the results of \cite{RS1,RS2} and suggestions for future work.

Then Sec.~\ref{sec:vortices and instantons} compares the non-abelian connections introduced here to describe vortices with the symmetric instantons given in \cite{CD}. The explicit form of both connections is given and  evidence for a conjectured duality between the different vortex equations is discussed. Finally Sec.~\ref{sec:conclusion} summarises the paper and gives some future directions of research.
\section{Lie groups and Cartan connections}
\label{sec:notation and conventions}
\subsection{Group conventions}
To understand how to read Fig.~\ref{fig:unified picture summary figure} we need to explain what the notation  $\mathbb{H}^{1}_{\lambda}$ means. We are interested in the three Lie groups $SU(2)$, $SE_{2}$, and $SU(1,1)$. Respectively these are the group of determinant one $2\times 2$ unitary matrices, the component of the Euclidean group in two dimensions connected to the identity, and the group of $2\times 2$ pseudo-unitary matrices. For a more concrete realisation of the groups we take the generators to be
\begin{equation}
t_{0}=-\frac{i}{2}\begin{pmatrix}
1&0\\0&-1
\end{pmatrix}, \quad
t_{1}=-\frac{i}{2}\begin{pmatrix}
0&-\lambda\\1&0
\end{pmatrix}, \quad 
t_{2}=\frac{1}{2}\begin{pmatrix}
0&\lambda\\1&0
\end{pmatrix},
\end{equation} 
with the commutation relations
\begin{equation}
[t_{a},t_{b}]=C_{ab}^{\,\,\, c}t_{c} \qquad \text{with} \quad C_{01}^{\,\,\,2}=1, \quad C_{02}^{\,\,\,1}=-1,\quad C_{12}^{\,\,\,0}=-\lambda,
\end{equation}
and all others vanishing. We use the notation $\mathbb{H}^{1}_{\lambda}$ for this group\footnote{This group is equivalent to the group $G_{C}$ considered in \cite{CD}. In \cite{CD} the authors pick the generators $J_{a}=-t_{a}^{T}$ and $C=-\lambda$. The conventions here are chosen so that for $\lambda=-1$ they match those in \cite{RS1}.} with;
\begin{equation}
\mathbb{H}^{1}_{-1}=SU(2), \quad \mathbb{H}^{1}_{0}=SE_{2}, \quad \mathbb{H}^{1}_{1}=SU(1,1).
\end{equation}   

It is convenient to introduce the complex combinations 
\begin{equation}
t_{\pm}=t_{1}\pm it_{2} \label{complex combination generators}
\end{equation}
which satisfy
\begin{equation}
[t_{0},t_{\pm}]=\mp it_{\pm}, \quad [t_{+},t_{-}]=2i\lambda t_{0}.
\end{equation}

The first of these can be interpreted as $t_{0}$ defining a complex structure on its complement such that $(t_{+}) t_{-}$ is (anti)-holomorphic. 

The group has inverse metric
\begin{equation}
g^{ab}=\text{diag}(-\lambda,1,1) \label{eq:inverse metric on group}
\end{equation}
which is used to raise and lower group indices, and is degenerate in the $\lambda=0$ case.

As a submanifold of $\C^{2}$, $\mathbb{H}_{\lambda}^{1}$ is defined as
\begin{equation}
\mathbb{H}_{\lambda}^{1}=\{(z_{1},z_{2})\in\C^2 \vert \vert z_{1}\vert-\lambda \vert z_{2}\vert=1\},
\end{equation}
the signature of the submanifold depends on the sign of $\lambda$. The complex coordinates $(z_{1},z_{2})$ parametrise a matrix $h\in\mathbb{H}^{1}_{\lambda}$ through
\begin{equation}
h=\begin{pmatrix}
z_{1}&\lambda \bar{z}_{2}\\
z_{2}&\bar{z}_{1}
\end{pmatrix}.
\end{equation}

As $\mathbb{H}^{1}_{\lambda}$ is a Lie group there is a real, left-invariant Maurer-Cartan one-form encoding the geometry,
\begin{equation}
h^{-1}dh=\sigma^{0}t_{0}+\sigma^{1}t_{1}+\sigma^{2}t_{2}. \label{eq:MC one-form}
\end{equation}
We say they encode the geometry as they satisfy the structure equations
\begin{equation}
d\sigma^{a}=-\frac{1}{2}C_{bc}^{\;\;\;a}\lf{b}\wedge\lf{c},
\end{equation}
where the $C_{bc}^{\;\;\;a}$ are the structure constants of the group. As a precursor to what we will consider later, observe that $h^{-1}dh$ can be viewed as a $\mathbb{H}^{1}_{\lambda}$ valued connection on $\mathbb{H}^{1}_{\lambda}$, whose flatness is equivalent to the structure equations. Encoding the geometry of manifolds in terms of flat connections is a central theme of Cartan geometry, see the textbook~\cite{Sharpe} and the PhD thesis~\cite{Wise} for a general discussion of Cartan geometry.

Throughout it is convenient to work with the complex combinations
\begin{equation}
\sigma=\lf{1}+i\lf{2}, \qquad \bar{\sigma}=\lf{1}-i\lf{2},
\end{equation}
which obey
\begin{equation}
d\sigma=i\sigma\wedge \lf{0}, \qquad d\lf{0}=\frac{i\lambda}{2}\sigma\wedge\bar{\sigma}. \label{eq:3d structure equations}
\end{equation}
In terms of the complex coordinates the left invariant one-forms have the explicit expressions
\begin{equation}
\sigma=2i\left(z_{1}dz_{2}-z_{2}dz_{1}\right), \quad \lf{0}=i\left(\bar{z}_{1}dz_{1}-\lambda\bar{z}_{2}dz_{2}-z_{1}d\bar{z}_{1}+\lambda z_{2}d\bar{z}_{2}\right). \label{eq:left invariant frame}
\end{equation}
In terms of the left invariant one-forms the metric and orientation\footnote{The slightly unconventional ordering is so that it makes contact with the volume form on $\R^{3}$ after stereographic projection in the $\lambda=1,-1$ cases.} are
\begin{align}
ds^{2}&=\frac{1}{4}\left(-\frac{1}{\lambda}\left(\sigma^{0}\right)^{2}+\left(\sigma^{1}\right)^{2}+\left(\sigma^{2}\right)^{2}\right),\label{eq:metric using left-invariant one-forms}\\
\text{Vol}_{\mathbb{H}^{1}_{\lambda}}&=\frac{1}{8}\sigma^{1}\wedge\sigma^{0}\wedge\sigma^{2}\label{eq:orientation using left-invariant one-forms}.
\end{align}
The metric is singular in the $\lambda=0$ case but the only problem due to this is that we are unable to construct zero modes on $\mathbb{H}_{0}^{1}$ in Sec.~\ref{sec:magnetic modes} .

The dual left-invariant vector fields, $X_{a}$, generate the right action $h\to h t_{a}$ and have the commutators
\begin{equation}
[X_{a},X_{b}]=C_{ab}^{\;\;\;c}X_{c}.
\end{equation}
In terms of the combinations
\begin{equation}
X_{\pm}=X_{1}\pm iX_{2}
\end{equation}
we have
\begin{equation}
[X_{0},X_{\pm}]=\mp iX_{\pm}, \qquad [X_{+},X_{-}]=2\lambda iX_{0}.
\end{equation}
In terms of the complex coordinates the left invariant vector fields take the form
\begin{align}
X_{0}&=-\frac{i}{2}\left(z_{1}\partial_{1}+z_{2}\partial_{2}-\bar{z}_{1}\bar{\partial}_{1}-\bar{z}_{2}\bar{\partial}_{2} \right),\\
X_{-}&=-i\left(\bar{z}_{1}\partial_{2}+\lambda \bar{z}_{2}\partial_{1}\right),\\
X_{+}&=\overline{X_{-}}.,
\end{align}
where we have used $\partial_{i}=\frac{\partial}{\partial z_{i}}$. The only non zero pairing are
\begin{equation}
\lf{0}(X_{0})=1, \qquad \sigma(X_{-})=\bar{\sigma}(X_{+})=2.
\end{equation}

A key feature  of $\mathbb{H}^{1}_{\lambda}$ is that it is a circle fibration over a constant curvature Riemann surface $M$. For $\lambda=-1$ this is the familiar Hopf fibration, while in the other cases we have a trivial bundle. The projection is
\begin{equation}
\pi:\mathbb{H}^{1}_{\lambda}\to M, \qquad h\mapsto z=\frac{z_{2}}{z_{1}}. 
\end{equation}
There is also the familiar section, local when $\lambda=-1$ but global otherwise,
\begin{equation}
s:z\mapsto \frac{1}{\sqrt{1-\lambda\vert z\vert}}\begin{pmatrix}
1&\lambda\bar{z}\\
z&1
\end{pmatrix}.\label{eq:bundle section}
\end{equation}
This enables us to relate our Maurer-Cartan one form to the Cartan connection for the Riemann surface $M$.

The group manifold $\mathbb{H}^{1}_{\lambda}$ is not simply connected when $\lambda=0,1$. This is because topologically $\mathbb{H}^{1}_{-1}=S^{3},\,\mathbb{H}^{1}_{0}=\R^{2}\times S^{1}$, and  $\mathbb{H}^{1}_{1}=H^{2}\times S^{1}$. The generator of the fundamental group is the curve
\begin{equation}
\gamma=\{e^{\varphi t_{0}}\in \mathbb{H}^{1}_{\lambda}\vert \varphi\in[0,4\pi)\}.
\end{equation}
This curve is contractable when $\lambda=-1$ since $\pi_{1}(\mathbb{H}^{1}_{\lambda})=\pi_{1}(S^{3})=0$.

When it is non-trivial, only flat connections with a prescribed holonomy around $\gamma$ can be globally trivialised. We will encounter this constraint when discussing vortices on the group manifold in Sec.~\ref{sec:vortex configurations}.

\subsection{Two dimensional geometry}
On a Riemann surface $M$ with constant Gauss curvature $K$ we work in local complex coordinates $z$. As we are considering $M$ to be either $S^{2}, H^{2}$ or $\R^{2}$, $z$ is a global coordinate except on $S^{2}$. The Riemann surface has metric\footnote{Note that with this choice of metric $\R^{2}$ with $K=0$ has $ds^{2}=4dzd\bar{z}$. This may seem an unusual choice but it facilitates the same language to be used for all three surfaces.}
\begin{equation}
ds^{2}=\Omega dz d\bar{z}=\frac{4}{\left(1+K\vert z\vert^{2}\right)^{2}}dzd\bar{z}.
\end{equation}

This metric admits the (local) complexified coframe
\begin{equation}
e=\frac{2dz}{1+K\vert z\vert^{2}}. \label{eq:complexified local frame}
\end{equation}
The geometry of the Riemann surface is encoded in the structure and Gauss equations,
\begin{align}
de-ie\wedge\Gamma&=0 \label{eq:structure eq}\\
d\Gamma=\mathcal{R}&=\frac{i}{2}Ke\wedge\bar{e} \label{eq:Gauss eq}
\end{align}
where 
\begin{equation}
\Gamma=iK\frac{zd\bar{z}-\bar{z}dz}{1+K\vert z\vert^{2}}\label{eq:spin connection}
\end{equation}
is the spin connection one-form and $\mathcal{R}$ is the curvature two-form.
When we have two Riemann surfaces we denote the one with Gauss curvature $K$ by $M$ and the one with Gauss curvature $K_{0}$ by $M_{0}$. There are corresponding $0$ subscripts on the coframe fields, spin connection, and curvature two-form. 

There are two results worth noting here. The first is that the structure and Gauss equations can be interpreted as the flatness of a Cartan connection, $\hat{A}$ . The second relates $\hat{A}$ to the Maurer-Cartan one-form on $\mathbb{H}^{1}_{\lambda}$ when $\lambda=-K$.

\begin{prop} \label{prop:surface geometry flat connection}
The structure and Gauss equations, \eqref{eq:structure eq} and \eqref{eq:Gauss eq}, for the frame, \eqref{eq:complexified local frame}, and spin connection \eqref{eq:spin connection}, are equivalent to the flatness of the $\text{Lie}(\mathbb{H}^{1}_{\lambda})$ valued connection
\begin{equation}
\hat{A}=-\Gamma t_{0}+\frac{i}{2}\left(e t_{-}-\bar{e}t_{+}\right). \label{eq:Surface flat cartan connection}
\end{equation}
where $K=-\lambda$ is the Gauss curvature.
\end{prop}

The proof is a straightforward computation of the curvature of $\hat{A}$.
\begin{proof}
The curvature of $\hat{A}$ is 
\begin{align}
F_{\hat{A}}	&=d\hat{A}+\frac{1}{2}[\hat{A},\hat{A}],\\
			&=-\left(\mathcal{R}+\lambda\frac{i}{2}e\wedge\bar{e}\right)t_{0}+\frac{i}{2}\left(de-i\Gamma\wedge e\right)t_{-}-\frac{i}{2}\left(d\bar{e}+i\Gamma\wedge \bar{e}\right)t_{+}.
\end{align}
The vanishing of the coefficient of $t_{0}$ is equivalent to the Gauss equation, \eqref{eq:Gauss eq} with curvature $K=-\lambda$, and the vanishing of the $t_{\pm}$ coefficients  is equivalent to the structure equations, \eqref{eq:structure eq}.
\end{proof}

\begin{prop} \label{prop:surface connection from MC}
Using the local section \eqref{eq:bundle section} the Cartan connection $\hat{A}$ on $M$, \eqref{eq:Surface flat cartan connection}, is trivialised as the pullback of the Maurer-Cartan one-form, \eqref{eq:MC one-form},
\begin{equation}
\hat{A}=s^{*}\left(h^{-1}dh\right).
\end{equation}
\end{prop}
\begin{proof}
To prove use the explicit expression for the section, \eqref{eq:bundle section}, to compute $s^{*}\sigma=ie$ and $s^{*}\sigma^{0}=-\Gamma$. Then directly computing the pullback of Eq.~\eqref{eq:MC one-form} leads to $\hat{A}$.
\end{proof}

Going the other way round $h^{-1}dh$ can be expressed in terms of $\pi^{*}\hat{A}$. Unfortunately, this is only true up to a singular gauge transformation as
\begin{equation}
\pi^{*} e=-i\frac{\bar{z}_{1}}{z_{1}}\sigma, \qquad \pi^{*}\Gamma =-\sigma^{0}+id\ln\left(\frac{z_{1}}{\bar{z}_{1}}\right).
\end{equation}

These propositions convey a key concept of this work; we encode equations on a Riemann surface as flat connections, and relate these to the Maurer-Cartan one-form on $\mathbb{H}^{1}_{\lambda}$.

\section{Integrable vortex equations}
\label{sec:vortex equations}
\subsection{Vortices on Riemann surfaces}
Turning now to vortex equations on Riemann surfaces. The vortex solutions of Eqs.~\eqref{covariant holomorphic} and \eqref{eq:hyperbolic vortex part 2} are known as hyperbolic vortices, since the equations are integrable on $H^{2}$. These hyperbolic vortex equations were generalised in \cite{Manton2} to give integrable vortex equations on a more general Riemann surface $M_{0}$. The hyperbolic case has a long history with the first solutions given in \cite{Witten1}, while the spherical case $M_{0}=S^{2}$ was first considered in \cite{Manton1,Popov}, but the language was unified in \cite{Manton2} to give five\footnote{Or should that be 6 equations, the $\lambda=\lambda_{0}=0$ case was not considered in \cite{Manton2} but was called the Laplace vortex equation in \cite{CD}. It corresponds to a covariantly holomorphic section $\phi$ and a flat connection $a$. These Laplace vortices fit into the Cartan framework given here.} integrable vortex equations. The equations can be written down on any Riemann surface but the integrability relies on having constant curvature.

The vortex equations involve two parameters, suggestively called $\lambda_{0}$ and $\lambda$, and describe a pair $(\phi,a)$ of a connection $a$ and a section $\phi$ of a complex line bundle over $M_{0}$. When $M_{0}$ is non-compact appropriate asymptotics, $\vert\phi\vert\to 1$ on $\partial M_{0}$, need to be imposed to ensure finite energy \cite{MS,JT}. 
\begin{dfn} \label{exotic vortex def}
A $(\lambda_{0},\lambda)$ vortex is a pair $(\phi,a)$ of a connection, $a$, and a section, $\phi$, of a complex line bundle over $M_{0}$ which satisfy the $(\lambda_{0},\lambda)$ vortex equations
\begin{equation}
\left(d\phi-ia\phi\right)\wedge e_{0}=0 \qquad F_{a}=da=\left(\lambda_{0}-\lambda|\phi|^{2}\right)\omega_{0}. \label{eq:Vortex equations}
\end{equation}
\end{dfn}
Here
\begin{equation}
\omega_{0}=\frac{i}{2}e_{0}\wedge\bar{e}_{0},
\end{equation} is the K\"{a}hler form on $M_{0}$.

From \cite{Manton2} solutions are given by rational maps $f:M_{0}\to M$ where $M_{0}, M$ are Riemann surfaces with constant Gauss curvature $-\lambda_{0},\lambda$ respectively. There is the following direct way to solve \eqref{eq:Vortex equations}, define
\begin{equation}
f^{*}e=\phi e_{0}, \qquad a=f^{*}\Gamma -\Gamma_{0}.
\end{equation}
Then pulling back Eq.~\eqref{eq:structure eq} by $f$ gives the first vortex equation, while pulling back the Gauss equation \eqref{eq:Gauss eq}, noting $f^{*}\left(e\wedge\bar{e}\right)=\vert \phi\vert e_{0}\wedge \bar{e}_{0}$ and using the Gauss equation on $M_{0}$, leads to the second vortex equation. The data of a holomorphic map between two constant curvature Riemann surfaces is thus all we need to construct a solution to the vortex equations.

Another way to show that is to reduce the vortex equations to the Lioville equation. Decomposing the Higgs field as $\phi=e^{u+i\chi}$ leads to a generalisation of the Taubes equation
\begin{equation}
-\frac{4}{\Omega_{0}}\partial_{z}\partial_{\bar{z}}h=\left(\lambda_{0}-\lambda e^{2h}\right)-\frac{2\pi}{\Omega_{0}}\sum_{r=1}^{N}\delta\left(z-Z_{r}\right),
\end{equation}
with $\Omega_{0}$ the conformal factor on $M_{0}$ and $Z_{r}\in\C$ the zeros of $\phi$. A scaling argument from \cite{Manton2} shows that there are five integrable cases:
\begin{itemize}
\item Hyperbolic vortices $\lambda_{0}=\lambda=1$,
\item Popov vortices $\lambda_{0}=\lambda=-1$,
\item Jackiw-Pi vortices $\lambda_{0}=0,\,\lambda=-1$,
\item Ambj\o rn-Olesen vortices $\lambda_{0}=1=-\lambda$,
\item Bradlow vortices $\lambda_{0}=1, \, \lambda=0$.
\end{itemize}
As noted above, the case of $\lambda_{0}=\lambda=0$, sometimes called Laplace vortices, could also be included. This naturally fits into the framework that we discuss here. Explicit expressions for $\phi$ and $a$ in terms of $f$ are given in \cite{Manton2}.

\subsection{Baptista metric}
The idea of interpreting vortices geometrically stems from \cite{MR} where the Higgs field  of a Hyperbolic vortex is represented as the ratio of conformal factors
\begin{equation}
\vert \phi\vert^{2}=\frac{f^{*}\Omega}{\Omega_{0}}\left\vert \frac{df}{dz}\right\vert^{2}.
\end{equation}
Then in \cite{Baptista} it was shown that a vortex defines a degenerate conical geometry on $M_{0}$, where the metric has conformal factor $\vert \phi\vert^{2}$ which is zero at the vortex centres. This idea was extended in \cite{Manton2} where it is referred to as the Baptista metric. In the integrable cases the Baptista metric is given by
\begin{equation}
ds^{2}_{B}=f^{*}ds^{2}=\vert\phi\vert^{2}ds^{2}_{0}.\label{eq:Baptista metric}
\end{equation}
This says that a vortex defines a degenerate frame on $M_{0}$, with the vortex equations forming part of the structure and Gauss equations for this frame.

It is important to be aware that since the Baptista metric has the conformal factor $|\phi|^{2}$ it is degenerate at the $N$, not necessarily distinct, zeros of the Higgs field.  As observed in \cite{Baptista} the Riemann curvature 2-form associated with the metric is extended to the zeros by adding delta function singularities
\begin{equation}
\mathcal{R}'=\mathcal{R}_{0}+F_{a}-2\pi\sum_{j=1}^{N}\delta_{Z_{j}}, \label{eq:singular curvature}
\end{equation}
where we use $\delta_{Z_{j}}$ for the two-form Dirac delta supported on the point $Z_{j}$.

This can be understood as the Baptista metric having a conical singularity with surplus angle $2\pi N_{j}$ at a zero of multiplicity $N_{j}$, with $N=\sum_{j} N_{j}$. The local geometry around the point $Z_{j}$ thus resembles a ruffled collar and is sometimes called an Elizabethan geometry.  For the case of Popov vortices on $S^{2}$, $a$ is a connection on a line bundle of even degree, $N=2n-2$  with $n=1, 2,\dots$ etc, and thus
\begin{equation}
\int_{S^{2}}F_{a}=4\pi n-4\pi.
\end{equation}
This is cancelled by the integral over the delta functions and thus
\begin{equation}
\int_{S^{2}}\mathcal{R}'=\int_{S^{2}}\mathcal{R}_{S^{2}}=4\pi,
\end{equation}
which can be interpreted as the Gauss-Bonnet theorem holding for $\mathcal{R}'$. This is in contrast to the pullback of the curvature two-form, $f^{*}\mathcal{R}_{S^{2}}$, which integrates to $4\pi n$ since in this case the map $f:S^{2}\to S^{2}$ has degree $n$ \cite{RS1,Manton2}.

For the other types of vortex we still have Equation \eqref{eq:singular curvature}, and $F_{a}$ still integrates to $2\pi N$, once the appropriate boundary conditions are taken in to account, which is again cancelled by the delta function contribution. However, as $H^{2}$ and $\R^{2}$ are non-compact the integrals of the curvature forms are not defined. One way to get around this is to work on compact manifolds covered by $H^{2}$ or $\R^{2}$, such as a Riemann surface of genus $g>1$ which is the quotient of $H^{2}$ by a Fuchsian group $\Lambda< SU(1,1)$. This complicates the story somewhat so we do not focus on it here.

This example shows that the spin connection of the degenerate frame $\phi e_{0}$, $\tilde{\Gamma}$ differs from the pulled back spin connection $f^{*}\Gamma$ by a contribution due to the zeros of $\phi$ and that this contribution is what leads to the singularities in $\mathcal{R}'$.

\subsection{Vortices as flat connections}
From the discussion of the Baptista metric and Proposition \ref{prop:surface geometry flat connection} it follows that the vortex equations are equivalent to the flatness of the non-abelian connection $f^{*}\hat{A}$. The key observation is that $f^{*}e=\phi e_{0}$ defines a degenerate frame from the vortex, and the structure and Gauss equations for this frame imply the vortex equations of Eq.~\eqref{eq:Vortex equations}.

\begin{cor}\label{cor:vortex as a flat connection}
Given the flat connection $\hat{A}$ defined in Eq.~\eqref{eq:Surface flat cartan connection}, and a holomorphic map $f:M_{0}\to M$, the flatness of
\begin{equation}
f^{*}\hat{A}=-\left(a+\Gamma_{\lambda_{0}}\right)t_{0}+\frac{i}{2}\left(\phi e_{\lambda_{0}}t_{-}-\bar{\phi}\bar{e}_{\lambda_{0}}t_{+}\right), \label{eq:Vortex flat connection 2d}
\end{equation}
is equivalent to the $(\lambda_{0},\lambda)$ vortex equations.
\end{cor}

Since $f$ is a rational function it can be written as $f(z)=f_{2}(z)/f_{1}(z)$ for two holomorphic functions $f_{1},f_{2}$. The Higgs field $\phi$, and thus the connection $f^{*}\hat{A}$ have zeros, and potentially singularities, determined by the $f_{i}$. To see this explicitly note that
\begin{equation}
f^{*}e=\phi e=2\frac{f_{2}' f_{1}-f_{1}' f_{2}}{\vert f_{1}\vert^{2}+K\vert f_{2}\vert^{2}}\frac{\bar{f}_{1}}{f_{1}}dz,
\end{equation} 
which has zeros at the ramification points of $f$, and singularities at the zeros of $f_{1}$. There are singularities when the coframe $e$ has a singularity, this happens when $M=S^{2}$ as the coframe is local in a patch with coordinate singularity at one of the poles of the sphere. Under $f$ the pre-images of the coordinate singularity of $e$ become the zeros of $f_{1}$. 

For example if $q$ is a zero of $f_{1}$ then near $q$ 
\begin{equation}
\phi e\sim A \frac{\bar{z}-\bar{q}}{z-q}dz,
\end{equation} 
for a constant $A$. 
When $\phi e$ has singularities they are inherited by the vortex Cartan connection $f^{*}\hat{A}$. Thus it is not really defined on the total space $\mathbb{H}^{1}_{\lambda}$ but on 
\begin{equation}
P=\mathbb{H}^{1}_{\lambda}\backslash \bigcup_{j}\pi^{-1}(q_{j}),\label{eq:total space with singularities}
\end{equation}
where the $q_{j}$ are the zeros of $f_{1}$. For the case of Popov vortices an extensive discussion of the singularities and their properties is given in \cite{RS1}.

In the language of Cartan connections our results so far are that $\hat{A}$ is a gauge potential for the Cartan connection describing the geometry of $M$, while $f^{*}\hat{A}$ is a gauge potential for a Cartan connection describing the deformed geometry on $M_{0}$ defined by the vortex $(\phi,a)$.

\section{Vortices on Lie Groups}
\label{sec:vortex configurations}
\subsection{Vortex configurations}
The next important actor in this story are vortex configurations on $\mathbb{H}^{1}_{\lambda_{0}}$. These are the generalisation of vortices to three dimensional group manifolds. Unlike vortices they do not involve sections and connections but give a way of writing the vortex equations in terms of one-forms and complex functions on the total space of the bundle. 

Before discussing vortex configurations we need to understand the equivariant functions on $\mathbb{H}^{1}_{\lambda}$.

\begin{dfn}
\label{equivariant function definition}
The space of equivariant functions over $\mathbb{H}^{1}_{\lambda}$ is defined to be
\begin{equation}
C^{\infty}\left(\mathbb{H}^{1}_{\lambda},\C\right)_{N}=\{F:\mathbb{H}^{1}_{\lambda}\to \C|2iX_{0}F=NF\}, \quad N\in \Z. 
\end{equation}
\end{dfn}

In \cite{RS1} the discussion of equivariant functions followed that in \cite{JS,JS2} with $N\in \N^{0}$. This is because for $SU(2)$, equivariant functions are functions on the Lens space $S^{3}/\Z_{N}$ and are related to sections of the hyperplane bundle. In general it is not a priori clear that we need to impose the same restriction that the degree is a non-negative integer. In practice we only encounter equivariant functions constructed from holomorphic polynomials in $z_{1},z_{2}$ and for these $N\in\N^{0}$. 

A short computation using the local section $s$, \eqref{eq:bundle section}, results in the following commutative diagram
\begin{equation}
\xymatrix{C^{\infty}(\mathbb{H}^{1}_{\lambda},\C)_{N}\ar[r]^{X_{+}}\ar[d]_{s^{*}}& C^{\infty}(\mathbb{H}^{1}_{\lambda},\C)_{N+2} \ar@{->}[d]^{s^{*}}\\  C^{\infty}(M_{\lambda})\ar@{->}[r]_{i(q\bar{\partial}-\lambda\frac{N}{2}z)}& C^{\infty}(M_{\lambda})} \label{cd:X plus holomorphicity general}
\end{equation}
where $q=(1-\lambda|z|^{2})$. The fact that for $\lambda=0$ the vector field $X_{+}$ is the lift of $\bar{\partial}$ is not particularly surprising as in this case $q=1$, $X_{+}=iz_{1}\bar{\partial}_{2}$ and the section $s$ identifies $z_{2}$ with $z$. 

As we will be dealing with two, potentially different, group manifolds $\mathbb{H}^{1}_{\lambda_{0}}$ and $\mathbb{H}^{1}_{\lambda}$ we require separate notation for the geometric objects on the two spaces. We use the notation $s_{a}, \sigma^{a}, X_{a}$ for the generators, left-invariant one-forms and left-invariant vector fields respectively of the source group three manifold, $\mathbb{H}^{1}_{\lambda_{0}}$, and $t_{a},\tau^{a},Y_{a}$ for the same objects on the target group three manifold, $\mathbb{H}^{1}_{\lambda}$. Due to the details of our construction the Cartan connection is always valued in the Lie algebra of the target group.

\begin{dfn} Let $A$ be a one form on $\mathbb{H}^{1}_{\lambda_{0}}$ and $\Phi:\mathbb{H}^{1}_{\lambda_{0}}\to \C$ a complex function. We call the pair $(\Phi,A)$ a vortex configuration if the vortex equations,
\begin{equation}
\left(d\Phi+iA\Phi\right)\wedge\sigma=0\qquad F_{A}=\frac{i}{2}\left(\lambda_{0}-\lambda|\Phi|^{2}\right)\bar{\sigma}\wedge\sigma, \label{eq:3D vortex equations}
\end{equation}
with $F_{A}=dA$, are satisfied.
\end{dfn}

In the $\lambda=\lambda_{0}=-1$ case of \cite{RS1} normalisation and equivariance conditions were included in the definition of vortex configurations. However, in the other cases $A(X_{0})$ is not necessarily an integer  so we do not include a normalisation condition. When we construct examples of vortex configurations they will be normalised.  Equivariance conditions, when $A$ is normalised, follow from the vortex equations, \eqref{eq:3D vortex equations} and Cartan's identity. They are
\begin{align}
\mathcal{L}_{X_{0}}\Phi&=-iA(X_{0})\Phi, \label{eq:higgs field eqivariance} \\
\mathcal{L}_{X_{0}}A&=dA(X_{0}). \label{eq:connection equivariance}
\end{align} 

These vortex equations possess the $U(1)$ gauge invariance
\begin{equation}
(\Phi,A)\mapsto (e^{-i\beta }\Phi, A+d\beta), \qquad \beta\in C^{\infty}(\mathbb{H}^{1}_{\lambda_{0}}). \label{eq:vortex config gauge transform}
\end{equation}

The three dimensional vortex equations in \eqref{eq:3D vortex equations} have clear similarities to the vortex equations in \eqref{eq:Vortex equations}. The left-invariant one forms are the analogue of the complexified coframe. The precise relationship will be established shortly, it relies on the relationship between $e,\bar{e},\Gamma$ and the $\sigma^{i}$ under pullback by $\pi$ and $s$.

We know come to the central theorem of the paper, a method for constructing vortex configurations from bundle maps.

\begin{thm}
\label{vortex config to bundle mp and back}
A vortex configuration on $\mathbb{H}^{1}_{\lambda_{0}}$ determines a gauge potential for a flat $\mathbb{H}^{1}_{\lambda}$ connection of the form
\begin{equation}
\mathcal{A}=\left(A+\lf{0}\right)t_{0}+\frac{1}{2}\Phi\sigma t_{-}+\frac{1}{2}\bar{\Phi}\bar{\sigma}t_{+}. \label{eq:potential in vortex gauge}
\end{equation}
Conversely, a flat $\mathbb{H}^{1}_{\lambda}$ connection $\mathcal{A}$ on $\mathbb{H}^{1}_{\lambda_{0}}$ such that
\begin{equation}
\mathcal{A}(X_{0})=pt_{0}, \qquad \mathcal{A}(X_{-})=\alpha t_{0}+\Phi t_{-}, \label{eq:coefficient conditions}
\end{equation}
for functions $p:\mathbb{H}^{1}_{\lambda_{0}}\to \R$ and $\alpha, \Phi:\mathbb{H}^{1}_{\lambda_{0}}\to \C$ determines a vortex configuration $(\Phi,A)$ through the expansion \eqref{eq:potential in vortex gauge}.

Given $\mathcal{A}$, a gauge potential for a flat $Lie\left(\mathbb{H}^{1}_{\lambda}\right)$ connection on $\mathbb{H}^{1}_{\lambda_{0}}$ of the form \eqref{eq:potential in vortex gauge} which satisfies
\begin{equation}
\int_{\gamma}\mathcal{A}=2\pi n t_{0}, \label{eq:holonomy condition}
\end{equation}
for $n\in \Z$, it can be trivialised as $U^{-1}dU$ for a bundle map $U:\mathbb{H}^{1}_{\lambda_{0}}\to \mathbb{H}^{1}_{\lambda}$ which covers a holomorphic map $f: M_{0}\to M$. Without loss of generality $U$ can be taken to have the form
\begin{equation}
U:(z^{1},z^{2})\mapsto \frac{1}{\sqrt{|F_{1}|^{2}-\lambda |F_{2}|^{2}}}\begin{pmatrix}
F_{1}&\lambda \bar{F}_{2}\\
F_{2}& \bar{F}_{1}
\end{pmatrix} \label{eq:U parameterisation}
\end{equation}
with $F_{i}:\mathbb{H}^{1}_{\lambda_{0}}\to \C$ where $|F_{1}|^{2}>\lambda|F_{2}|^{2}$.

The vortex configuration can be extracted from the bundle map as
\begin{equation}
\Phi \sigma=U^{*}\tau, \qquad A=U^{*}\tau^{0}-\lf{0}. \label{eq:Higgs and connection from geometry}
\end{equation}
\end{thm}

This result generalises Theorem 3.2 from \cite{RS1} and Theorem 3.2 from \cite{RS2}, the result when $\lambda_{0}=\lambda=-1$ or $1$ respectively. 
Note that when $\lambda\neq -1$ we can assume that $F_{1}\neq 0$ since $|F_{1}|^{2}>\lambda|F_{2}|^{2}$. However when $\lambda=-1$ this assumption is not valid. This relates to the fact that when $\lambda=-1$ the map $f=s^{*}\left(\frac{F_{2}}{F_{1}}\right)$, which $U$ covers, can have poles since it is a map $f:M_{0}\to S^{2}$. 

\begin{proof}
Given a  $\mathbb{H}^{1}_{\lambda}$ connection on $\mathbb{H}^{1}_{\lambda_{0}}$ in the vortex gauge, \eqref{eq:potential in vortex gauge}, the flatness condition $d\mathcal{A}+\mathcal{A}\wedge \mathcal{A}=0$ is equivalent to
\begin{equation}
\left(d\left(\Phi\sigma\right)+i\left(A+\sigma^{0}\right)\Phi\right)=0\qquad dA=\frac{i}{2}\left(\lambda_{0}-\lambda|\Phi|^{2}\right)\bar{\sigma}\wedge\sigma.
\end{equation}
Using equation \eqref{eq:3d structure equations} these are seen to be equivalent to the vortex equations, \eqref{eq:3D vortex equations}.

For the converse expand the flat $\mathbb{H}^{1}_{\lambda}$ connection, $\mathcal{A}$, on $\mathbb{H}^{1}_{\lambda_{0}}$  in terms of the generators, $t_{0}, t_{+}, t_{-}$. The coefficients are linear combinations of $\lf{0}, \sigma$ and $ \bar{\sigma}$, as they form a basis of the cotangent space of $\mathbb{H}^{1}_{\lambda_{0}}$. Imposing the conditions in Equation \eqref{eq:coefficient conditions} leads to the gauge potential $\mathcal{A}$ being in the vortex form, Equation \eqref{eq:potential in vortex gauge}, with Higgs field $\Phi$, and abelian gauge potential
\begin{equation}
A=\left(p-1\right)\lf{0}+\frac{1}{2}\left(\alpha \sigma+\bar{\alpha}\bar{\sigma}\right).
\end{equation}
The same calculation as above then gives that the vortex equations, \eqref{eq:3D vortex equations}, are satisfied.

To globally trivialise a flat connection $\mathcal{A}$ on $\mathbb{H}^{1}_{\lambda_{0}}$ in terms of $U:\mathbb{H}^{1}_{\lambda_{0}}\to \mathbb{H}^{1}_{\lambda}$ as $\mathcal{A}=U^{-1}dU$, its path-ordered exponential must be path independent. If this is the case $U$ can be constructed explicitly from $\mathcal{P} \exp\left(\int_{\tilde{\gamma}}\mathcal{A}\right)$ along any path $\tilde{\gamma}$, starting at a fixed but arbitrary base point, \cite{BM}.

As $\mathcal{A}$ is a flat connection the non-abelian Stokes theorem implies that the path-ordered exponential is path independent for contractible paths. The conditions in \eqref{eq:coefficient conditions} ensure that the path-ordered exponential of $\mathcal{A}$ along $\gamma$ coincides with the exponential of the ordinary integral. Then \eqref{eq:holonomy condition} implies that
\begin{equation}
\mathcal{P} \exp\left(\int_{\gamma}\mathcal{A}\right)=\mathbb{I}.
\end{equation}

Flatness of the connection $\mathcal{A}$ and the non-abelian Stokes theorem then combine to give that the path-ordered exponential of $\mathcal{A}$ along any closed curve in $\mathbb{H}^{1}_{\lambda_{0}}$ is the identity. This gives the path independence of the path-ordered exponential.

The final part of the proof is to show that for $\mathcal{A}=U^{-1}dU$ satisfying \eqref{eq:coefficient conditions} that $U$ is a bundle map covering a holomorphic function $M_{0}\to M$. The first condition in \eqref{eq:coefficient conditions} becomes
\begin{equation}
X_{0}U=p U t_{0}, \label{eq:U bundle map condition}
\end{equation}
with $p:\mathbb{H}^{1}_{\lambda_{0}}\to \R$. This is just the infinitesimal statement that $U$ maps the fibres of $\mathbb{H}^{1}_{\lambda_{0}}\to M_{0}$ to the fibres of $\mathbb{H}^{1}_{\lambda}\to M$, in other words that $U$ is a bundle map.

Complex conjugation of the second condition in \eqref{eq:coefficient conditions} implies that 
\begin{equation}
U^{-1}X_{+}U=\bar{\alpha}t_{0}+\bar{\Phi}t_{+}. \label{eq:conjugation of coeff condition}
\end{equation}
Now apply 
\begin{equation}
U^{-1}dU=U^{*}\tau^{0}t_{0}+\frac{1}{2}U^{*}\tau t_{-}+\frac{1}{2}U^{*}\bar{\tau}t_{+}
\end{equation} to $X_{+}$, the condition in \eqref{eq:conjugation of coeff condition} is thus equivalent to 
\begin{equation}
U^{*}\tau(X_{+})=0. \label{eq:U covers holomorphic}
\end{equation}
We now need to show that this is equivalent to $U$ covering a holomorphic map.

Using the parameterisation of $U$ in terms of the functions $F_{i}$ defined in \eqref{eq:U parameterisation} we see that Equation \eqref{eq:U bundle map condition} becomes
\begin{equation}
X_{0}\left(\frac{F_{i}}{\sqrt{|F_{1}|^{2}-\lambda |F_{2}|^{2}}}\right)=\frac{i}{2}p\left(\frac{F_{i}}{\sqrt{|F_{1}|^{2}-\lambda |F_{2}|^{2}}}\right).
\end{equation}

From this it follows that the map $\pi\circ U=\frac{F_{2}}{F_{1}}$ has equivariant degree zero, from Definition \ref{equivariant function definition}, and that $U$ covers 
\begin{equation}
f=s^{*}\left(\frac{F_{2}}{F_{1}}\right):M_{0}\to M.
\end{equation}

Using \eqref{cd:X plus holomorphicity general} for $\frac{F_{2}}{F_{1}}$ we find that $f$ being holomorphic is equivalent to 
\begin{equation}
X_{+}\left(\frac{F_{2}}{F_{1}}\right)=0.
\end{equation}
Returning to \eqref{eq:U covers holomorphic} use \eqref{eq:left invariant frame} to see that
\begin{equation}
U^{*}\tau(X_{+})=\frac{2i}{|F_{1}|^{2}-\lambda|F_{2}|^{2}}\left(F_{1}X_{+}F_{2}-F_{2}X_{+}F_{1}\right)=\frac{2i}{|F_{1}|^{2}-\lambda|F_{2}|^{2}}F_{1}^{2}X_{+}\frac{F_{2}}{F_{1}},
\end{equation}
with the last equality holding away from the zeros of $F_{1}$. Thus the condition \eqref{eq:U covers holomorphic} is equivalent to $f=s^{*}\left(\frac{F_{2}}{F_{1}}\right)$ being holomorphic away from the zeros of $F_{1}$. This means that for $\lambda\neq -1$ the result has been established. For the $\lambda=-1$ case $f$ will have poles at the zeros of $F_{1}$ and is a holomorphic map $M_{0}\to S^{2}$. 
\end{proof}

At the level of the bundle map $U$ the $U(1)$ gauge invariance from \eqref{eq:vortex config gauge transform} becomes
\begin{equation}
U\mapsto \tilde{U}=Ue^{\beta t_{0}}, \qquad \beta\in C^{\infty}\left(\mathbb{H}^{1}_{\lambda_{0}}\right).
\end{equation}
This defines a new trivialisation with the same $f$ as $U$. The connection $\tilde{U}^{-1}d\tilde{U}$ differs from $\mathcal{A}=U^{-1}dU$ by the gauge transformation in \eqref{eq:vortex config gauge transform}.

Notice that when $\lambda\neq -1$ a vortex configuration can always be constructed from a given holomorphic map $f:M_{0}\to M$ by choosing
\begin{equation}
F_{1}(z_{1},z_{2})=1,\qquad F_{2}(z_{1},z_{2})=f\left(\frac{z_{2}}{z_{1}}\right). \label{eq:trivial lift}
\end{equation}
This trivial lift results in a connection $\mathcal{A}$ that is constant along the fibres since $\mathcal{A}(X_{0})=0$ which implies $\mathcal{L}_{X_{0}}\mathcal{A}=0$. This trivial lift is a direct consequence of $\mathbb{H}^{1}_{1}$ and $\mathbb{H}^{1}_{0}$ being trivial bundles. The trivial lift was observed for the $\lambda=\lambda_{0}=1$ case in \cite{RS2}.

\subsection{Vortex configurations from vortices}
To construct vortex configurations from a non-trivial lift we follow the work of \cite{RS1,RS2} and lift vortices from $M_{0}$ to $\mathbb{H}^{1}_{\lambda_{0}}$. The idea behind this is that a vortex is given by a rational function $f=f_{2}/f_{1}$. We can then take the lift  $f_{2}/f_{1}=F_{2}/F_{1}$. This lift is non-trivial since the functions $F_{1},F_{2}$ have a non-trivial equivariant degree.

For this lift the Higgs field $\Phi$ has equivariant degree $2N-2$ in the sense of Definition~\ref{equivariant function definition}, with $N$ the integer equivariant degree of $F_{1}$ and $F_{2}$ determined from the rational function $f$. The specific vortex number depends on the type of vortex. Lifting Popov and hyperbolic vortices has been  carried out previously in \cite{RS1} and \cite{RS2} respectively. 

This leads us to the following Corollary of Theorem~\ref{vortex config to bundle mp and back}.
\begin{cor}
\label{cor:vortex config finite degree}
Let $U:\mathbb{H}^{1}_{\lambda_{0}}\to \mathbb{H}^{1}_{\lambda}$ be the bundle map from \eqref{eq:U parameterisation} with 
\begin{equation}
2iX_{0}F_{1}=2iX_{0}F_{2}=N\in\Z.
\end{equation}
The vortex configuration $(\Phi,A)$ constructed from the connection $\mathcal{A}=U^{-1}dU$ through Theorem \ref{vortex config to bundle mp and back} has a gauge field which satisfies the normalisation condition
\begin{equation}
A(X_{0})=N-1,
\end{equation}
and a Higgs field of equivariant degree $2N-2$. In terms of $F_{1},F_{2}$ the vortex configuration is expressed as
\begin{equation}
\Phi=\frac{F_{1}\partial_{2}F_{2}-F_{2}\partial_{2}F_{1}}{z_{1}\left(|F_{1}|^{2}+|F_{2}|^{2}\right)},
\end{equation}
and 
\begin{equation}
A=\left(N-1\right)\lf{0}+\frac{i}{2}X_{-}\ln D^{2} \sigma -\frac{i}{2}X_{+}\ln D^{2} \bar{\sigma}.
\end{equation}
with $D^{2}=|F_{1}|^{2}+|F_{2}|^{2}$.
\end{cor}

To make this result more understandable we give two example; one for Jackiw-Pi vortices, and another for Ambj\o rn-Olesen vortices.

\subsubsection{Jackiw-Pi vortices}
In \cite{HY} it was shown that the Jackiw-Pi vortex equations on $\R^{2}$ with a finite number of zeros are solved by a rational map $f=\frac{p}{q}:\R^{2} \to S^{2}$ with $\text{deg}(p)<\text{deg}(q)$. For example a $2N$-vortex solution is given by
\begin{equation}
p(z)=\sum_{i=0}^{M}a_{i}z^{i}, \qquad q(z)=\sum_{i=0}^{N}b_{i}z^{i}, \quad M<N,
\end{equation}
with the understanding that $p$ and $q$ have no common factors, at least one of $a_{0},b_{0}$ are non-zero, and  $b_{N}\neq 0$. In this case we can write down the following homogeneous polynomials 
\begin{equation}
P(z_{1},z_{2})=\sum_{i=0}^{M}a_{i}z_{1}^{N-i}z_{2}^{i}, \qquad Q(z_{1},z_{2})=\sum_{i=0}^{N}b_{i}z_{1}^{N-i}z_{2}^{i}, \label{eq:JP vortex polynomials}
\end{equation}
which satisfy 
\begin{equation}
s^{*}\left(\frac{P}{Q}\right)=\frac{p}{q}.
\end{equation}

Examples of Jackiw-Pi vortices with $N=1$ and $N=2$, including plots of $|\phi|^{2}$, are given in \cite{Horvathy1, HZ}.

In Corollary \ref{cor:vortex config finite degree} taking
\begin{equation}
F_{1}(z_{1},z_{2})=Q(z_{1},z_{2}), \qquad F_{2}(z_{1},z_{2})=P(z_{1},z_{2}),
\end{equation}
with $P,Q$ given in \eqref{eq:JP vortex polynomials} so that $2iX_{0}P=2iX_{0}Q=N$, defines a $\lambda_{0}=0,\lambda=-1$ vortex configuration.

In \cite{Manton2} the case of Jackiw-Pi vortices on the torus is discussed, there the map $f$ is a doubly periodic elliptic function. As a vortex on the torus it has a finite vortex number, $2N$ where $N$ is the number of poles of $f$. However, as a vortex on $\R^{2}$ it has an infinite number of zeros. The torus is obtained from $\R^{2}$ by quotienting with discrete subgroup of $SE_{2}$ and demanding that the zeros of the Higgs field on $\R^{2}$ are periodic under this subgroup and there are $2N$ of them in the principal domain. The only way to lift these vortices seems to be via the trivial lift \eqref{eq:trivial lift}. 

The most popular example of a Jackiw-Pi vortex on $\R^{2}$ \cite{HY, Horvathy1, HZ} is the axially symmetric case constructed from the rational function
\begin{equation}
f=\frac{1}{z^{N}}.
\end{equation}
For this choice of $f$ the $F_{i}$ are given by
\begin{equation}
F_{1}=z_{2}^{N},\qquad F_{2}=z_{1}^{N}.
\end{equation}
For this vortex the Higgs field of the vortex configuration is given by 
\begin{equation}
\Phi=-N\frac{z_{1}^{N-1}z_{2}^{N-1}}{|z_{1}|^{2N}+|z_{2}|^{2N}}.
\end{equation}
This can be explicitly seen to satisfy $2iX_{0}\Phi=(2N-2)\Phi$ and thus $\Phi$ is a degree $2N-2$ equivariant function.

\subsubsection{Ambj\o rn-Olesen vortices}
From \cite{Manton2} we know that  Ambjorn-Olsen vortices are constructed from a holomorphic map $f:H^{2}\to S^{2}$, subject to $|f(z)|\to 1$ as $|z|\to 1$. These maps can be expressed in terms of their $m$ zeros, $c_{1},\cdots , c_{m}$,  and $n$ poles, $d_{1},\cdots, d_{n}$,  as 
\begin{equation}
f(z)=\frac{f_{2}}{f_{1}}=\prod_{i=1}^{m} \left(\frac{z-c_{i}}{1-\bar{c}_{i}z}\right)\prod_{j=1}^{n} \left(\frac{1-\bar{d}_{j}z}{z-d_{j}}\right). \label{eq:AO map}
\end{equation}
To see this we use that the zeros and poles of $f$ define the Blaschke products
\begin{equation}
f_{2}=\prod_{i=1}^{m} \left(\frac{z-c_{i}}{1-\bar{c}_{i}z}\right), \quad f_{1}=\prod_{j=1}^{n} \left(\frac{z-d_{j}}{1-\bar{d}_{j}z}\right).
\end{equation}
The ratio of these Blaschke products $\frac{f_{2}}{f_{1}}$ has the same zeros and poles as $f$ and their ratio $\frac{f f_{1}}{f_{2}}$ is a holomorphic function with no zeros and no poles satisfying $|f(z)|=1$ for $|z|=1$. Liouville's theorem then gives that this ratio is a constant, $\mu \in \C$ such that $|\mu|=1$, multiplying $f$ by a constant does not change the vortex that we construct from $f$ so we can take $\mu=1$.

To make use of Corollary \ref{cor:vortex config finite degree} to construct a $\lambda_{0}=-\lambda=1$ vortex configuration take
\begin{equation}
F_{1}(z_{1},z_{2})=\prod_{i=1}^{n}\prod_{j=1}^{m}\left(z_{1}-\bar{c}_{i}z_{2}\right)\left(z_{2}-d_{j}z_{1}\right), \quad F_{2}(z_{1},z_{2})=\prod_{i=1}^{n}\prod_{j=1}^{m}\left(z_{2}-c_{i}z_{1}\right)\left(z_{1}-\bar{d}_{j}z_{2}\right). \label{eq:AO explicit maps}
\end{equation}
The equivariant degree is $N=m+n$.

The same procedure can be used to construct vortex configurations from a given Bradlow vortex.

\subsubsection{Lifts at the level of the connection}
Working at the level of the flat connections $\mathcal{A}$ and $f^{*}\hat{A}$ we can state the relationship between  vortex connections and vortices. To do this recall that if $\lambda=1$ then $f^{*}\hat{A}$ has singularities at the points $q_{j}$ such that $f_{1}(q_{j})=0$, and is thus defined on the space $P\subset \mathbb{H}^{1}_{\lambda_{0}}$ defined in Eq.~\eqref{eq:total space with singularities}. 

Note that for a function $g:M_{0}\to \C$ we can define the map
\begin{equation}
r_{g}:M_{\lambda_{0}}\backslash \{q_{i}\}\to \mathbb{H}^{1}_{\lambda}, \quad r_{g}=\begin{pmatrix}
\frac{\bar{g}}{|g|}& 0\\
0& \frac{g}{|g|}
\end{pmatrix},
\end{equation}
where the $q_{i}$ are the zeros of $g$.

Now using the section $s$ defined in \eqref{eq:bundle section}  we get the following corollary of Theorem \ref{vortex config to bundle mp and back} and Proposition \ref{prop:surface connection from MC}.
\begin{cor}\label{cor:vortex connection from vortex config}
For the bundle map $U$ in \eqref{eq:U parameterisation} covering the holomorphic map $f$, the gauge vortex connection $f^{*}\hat{A}$ from \eqref{eq:Vortex flat connection 2d} is related to $\mathcal{A}=U^{-1}dU$ through the, possibly singular, gauge transformation $r_{f_{1}}$, where $f_{1}=F_{1}\circ s$:
\begin{equation}
f^{*}\hat{A}=r_{f_{1}}^{-1}s^{*}\mathcal{A} r_{f_{1}}+r_{f_{1}}^{-1}dr_{f_{1}}.
\end{equation}
\end{cor}

The trivial lift in \eqref{eq:trivial lift} corresponds to $r_{f_{1}}=\mathbb{I}$, $f^{*}\hat{A}=s^{*}\mathcal{A}$. Again, this is only possible when $\lambda\neq 1$ so that $s^{*}\mathcal{A}$ and $f^{*}\hat{A}$ are both manifestly smooth.

There is a global Cartan geometry picture, in the sense of \cite{Sharpe,Wise}, corresponding to the vortex. In \cite{RS1} this is spelt out for the case of Popov vortices. All of the essential details carry forward to the more general case considered here. As with all of our discussion, the $\lambda=-1$ case is the most subtle and since that was dealt with in \cite{RS1} we do not give the details here.

\section{A comment on magnetic modes}
\label{sec:magnetic modes}
In the previous work \cite{RS1,RS2}, vortex configurations on the group manifold, either $SU(2)$ or $SU(1,1)$, were used to construct solutions to a twisted Dirac equation. These vortex magnetic modes were then pulled back to vortex magnetic modes on flat $\R^{3}$ or $\R^{2,1}$. For both Bradlow and Ambj\o rn-Olesen vortices the approach used in \cite{RS2} is applicable and the vortices lead to solutions of a twisted Dirac equation on $\R^{2,1}$.

As the argument follows through \textit{mutandis mutatis} we will not give all the details here but just state the result which generalises Theorem 5.3 from \cite{RS2}. The definition of a vortex zero mode is as follows
\begin{dfn}
A pair $(\Psi,A)$ of a spinor and a one-form on $SU(1,1)$ is called a vortex magnetic mode of the Dirac equation on $SU(1,1)$ if
\begin{equation}
\slashed{D}_{SU(1,1),A}\Psi=0, \qquad F_{A}=-\lambda\frac{4i}{\ell}\star\Psi^{\dagger}h^{-1}dh\Psi-\frac{1}{4}\sigma^{1}\wedge\sigma^{2},
\end{equation}
with $\star$ the Hodge star operator on $SU(1,1)$ with respect to the metric \eqref{eq:metric using left-invariant one-forms} and orientation \eqref{eq:orientation using left-invariant one-forms}.
\end{dfn}

Note that the signature of the metric here is opposite to that used in \cite{RS2}. The signature chosen here follows from the metric \eqref{eq:metric using left-invariant one-forms}, and enables a unified language for the different vortex equations. This changes does not effect the results and only changes a couple of signs in the computations.

\begin{thm}
Given $(\Phi,A)$ a vortex configuration on $SU(1,1)$, the pair
\begin{equation}
\Psi=\begin{pmatrix}
\Phi\\
0
\end{pmatrix}
,\qquad A'=A+\frac{3}{4}\sigma^{0},
\end{equation}
is a vortex magnetic mode on $SU(1,1)$.
\end{thm}
The proof from \cite{RS2} easily follows through once we note that the curvature due to $A'$ is
\begin{equation}
F_{A'}=F_{A}+\frac{3}{4}\sigma^{1}\wedge\sigma^{2}=\left(\lambda\vert\Phi\vert^{2}-\frac{1}{4}\right)\sigma^{1}\wedge\sigma^{2}.
\end{equation}
Finally Corollary 5.5 from \cite{RS2} still holds, leading to vortex magnetic modes on flat $\R^{2,1}$. This result relies on the properties of maps $H,G:\R^{2,1}\to SU(1,1)$, the inverse stereographic and gnomonic projections. To avoid getting bogged down in the details we do not give these maps here and refer the reader to \cite{RS2} for the details.

The astute reader will notice that we have so far said nothing about constructing vortex magnetic modes from Jackiw-Pi, or Laplace, vortices. This is because both Jackiw-Pi and Laplace vortices are related to the group $SE_{2}$ which does not possess a bi-invariant metric. In fact the metric from Eq.~\eqref{eq:inverse metric on group} is singular\footnote{Eq.~\eqref{eq:inverse metric on group} says that the inverse metric is degenerate which is equivalent to the metric being singular.} so we cannot construct a Dirac operator in the same way.  A potential approach to fixes this problem is to centrally extend $SE_{2}$ to the Nappi-Witten space \cite{NW} which has a Lorentzian metric. This central extension would not affect the construction in the other cases and the hope is that it will enable the construction of vortex magnetic modes from Jackiw-Pi vortices. This is a current direction of research.

\section{Vortices and instantons}
\label{sec:vortices and instantons}
The representation of vortices as non-abelian connections in three dimensions given here is an alternative to viewing vortices as symmetric, non-abelian instantons on flat $\R^{4}$. In \cite{CD} it was shown that all five of the integrable vortex equations can be constructed as the dimensional reduction of an appropriate (anti-) self dual Yang-Mills theory on $M_{0}\times N_{0}$, where $N_{0}$ is a Riemann surface with constant Gauss curvature $-K_{0}$. This construction is based on the general story of gauge fields which posses space time symmetries introduced in \cite{FM}. In \cite{RS1,RS2} it was observed that there is an interesting relationship between the gauge group of the Yang-Mills theory and the Lie algebra that the Cartan connection is valued in. In  short if the Cartan group is $\mathbb{H}^{1}_{\lambda}$ then the instanton gauge group is $\mathbb{H}^{1}_{-\lambda}$. 

%
%

In our conventions the construction from \cite{CD} considers instantons on the four manifold $\R^{4}\simeq M_{0}\times N_{0}$ with gauge group $\mathbb{H}^{1}_{-\lambda}$ that are equivariant with respect to the action of $\mathbb{H}^{1}_{-\lambda_{0}}$. This amounts to the instanton being independent of the $N_{0}$ factor and thus reducing to  a $(\lambda_{0},\lambda)$ vortex on $M_{\lambda_{0}}$. Explicitly the instanton is given by the $\mathbb{H}^{1}_{-\lambda}$-connection
\begin{equation}
\mathcal{A}_{\text{CD}}=-\left(a-\Gamma_{N_{0}}\right)t_{0}+\frac{i\phi}{2}\bar{e}_{N_{0}}t_{-}-\frac{i\bar{\phi}}{2}e_{-N_{0}}t_{+},
\end{equation}
with $(\phi,a)$ the $(\lambda_{0},\lambda)$ vortex on $M_{0}$ and $e_{N_{0}},\Gamma_{N_{0}}$ the complexified frame and spin connection on $N_{0}$. From Corollary \ref{cor:vortex as a flat connection} a $(\lambda_{0},\lambda)$ vortex is equivalent to a flat $\mathbb{H}^{1}_{\lambda}$-connection
\begin{equation}
A=-\left(a+\Gamma_{0}\right)t_{0}+\frac{i\phi}{2}e_{0}t_{-}-\frac{i\bar{\phi}}{2}\bar{e}_{0}t_{+}.
\end{equation}

It is interesting to contrast the two connections. $\mathcal{A}_{\text{CD}}$ is an anti- self dual connection on a conformally flat four-manifold while $A$ is a flat connection on a Riemann surface. This manifests itself in the fact that $A$ only depends on information on $M_{0}$, the vortex $(\phi,a)$ and the frame field and spin connection. On the other hand $\mathcal{A}_{\text{CD}}$ depends on information from both $M_{0}$ and $N_{0}$. Another difference that should be noted is that while we have used the same notation for the generators, $t_{0}, t_{\pm}$ they are not exactly the same, the key difference is in  $t_{+}$. For $\mathbb{H}^{1}_{\lambda}$ the explicit form of $t_{+}$ is
\begin{equation}
t_{+}=\begin{pmatrix}
0&i\lambda\\
0&0
\end{pmatrix}.
\end{equation}
This means the sign in $t_{+}$ is different for the two connections.

However, we know that the flat connection in two dimensions is related to a flat connection on the group manifold $\mathbb{H}^{1}_{\lambda_{0}}$ given by Eq.~\eqref{eq:potential in vortex gauge}. An immediate question is if there is a way to go directly between the instanton and the three dimensional Cartan connection. At the moment we only know how to pass between them by going through the vortex in two dimensions. There are definitely key differences in their construction with $A$ being constructed from the pullback of the left-invariant Maurer-Cartan one-form on $\mathbb{H}^{1}_{\lambda}$ while $\mathcal{A}_{\text{CD}}$, following the general construction in \cite{FM}, is constructed from  left-invariant data from $\mathbb{H}^{1}_{-\lambda_{0}}$ and right-invariant data from $\mathbb{H}^{1}_{-\lambda}$.

Finally consider the diagram 
\begin{equation}
\xymatrix{\mathbb{H}^{1}_{-\lambda_{0}}\times \mathbb{H}^{1}_{\lambda_{0}}\ar[d]^{\pi}\ar[r]^{U} & \mathbb{H}^{1}_{\lambda}\times \mathbb{H}^{1}_{-\lambda}\ar[d]^{\pi}\\ N_{0}\times M_{0}\ar[r]^{f} & M\times N}
\end{equation}
where $f:M_{0}\to M$ is the rational map defining a vortex and $U:\mathbb{H}^{1}_{\lambda_{0}}\to \mathbb{H}^{1}_{\lambda}$ is the bundle map encountered in Theorem \ref{vortex config to bundle mp and back}. From the instanton point of view $\mathbb{H}^{1}_{-\lambda_{0}}$ would be the symmetry group and $\mathbb{H}^{1}_{-\lambda}$ is the gauge group. This could be flipped round to 
\begin{equation}
\xymatrix{\mathbb{H}^{1}_{\lambda_{0}}\times \mathbb{H}^{1}_{-\lambda_{0}}\ar[d]^{\pi} & \mathbb{H}^{1}_{-\lambda}\times \mathbb{H}^{1}_{\lambda}\ar[d]^{\pi} \ar[l]^{V}\\ M_{0}\times N_{0} & N\times M\ar[l]^{g}}
\end{equation}
with $g:N\to N_{0}$ a rational map defining a vortex and $V:\mathbb{H}^{1}_{-\lambda}\to \mathbb{H}^{1}_{-\lambda_{0}}$ a bundle map. Now the instanton point of view has $\mathbb{H}^{1}_{\lambda}$ as the symmetry group and $\mathbb{H}^{1}_{\lambda_{0}}$ as the gauge group. 

This suggests that at the level of the groups there is a potential duality between the different vortex equations. This duality takes the $(\lambda_{0},\lambda)$ vortex equations to the $(-\lambda_{0},-\lambda)$ vortex equations. 

The Hyperbolic and Popov vortex equations are exchanged under this, as are the Bradlow and Jackiw-Pi vortices while both the Ambj\o rn-Olesen and Laplace vortex equations are ``self-dual'' in this sense.

\section{Conclusions and outlook}
\label{sec:conclusion}
This paper has considered the general problem of giving a geometric description of integrable abelian vortices as non-abelian flat connections. This provides a unified three dimensional interpretation of vortices, complementing the two dimensional metric geometry interpretation given by Baptista~\cite{Baptista}, and the four dimensional description of vortices as symmetric instantons \cite{CD}. 

The story is summarised in Fig.~\ref{fig:unified picture summary figure} where the most important maps, spaces and equations are given. This gives a unifying picture, generalising the work of \cite{RS1,RS2} to include all of the integrable vortex equations considered in \cite{Manton2}. As well as establishing the relationship between vortices and Cartan geometry we have also discussed proposals to construct solutions to massless Dirac equations from vortices.

A comparison between the Cartan connection picture and the instanton picture of vortices leads to some intriguing comparisons. Not least the fact that there seems to be a duality at the level of the groups where the $(\lambda_{0},\lambda)$ vortex equations are sent to the $(-\lambda_{0},-\lambda)$ vortex equations.

Recently the story of unified vortex equations has been extended. In \cite{walton_2021_exotic} non-abelian analogues of the integrable abelian vortex equations have been considered. While in \cite{GR21}, magnetic defects were added which preserved the integrability of the abelian vortex equations. It is unknown if in either of these cases there is still a geometric understanding of the vortices. Either in the Baptista sense, or in the Cartan geometry sense of this paper. Understanding these cases is a direction worth pursuing.

\paragraph{\bf Acknowledgements} The initial work for this project was carried out during my PhD. I acknowledge an EPSRC funded PhD studentship and thank my supervisor Bernd Schroers for introducing me to the topic and for many useful discussions. I also want to thank Chris Halcrow for providing comments on a draft of this paper.

%




\begin{thebibliography}{10}

\bibitem{RS1}
C.~ Ross and B.~J. Schroers.
\newblock {Magnetic Zero-Modes, Vortices and Cartan Geometry}.
\newblock {\em Lett. Math. Phys.}, 108(4):949--983, 2018.

\bibitem{RS2}
C.~ Ross and B.~J. Schroers.
\newblock {Hyperbolic vortices and Dirac fields in 2+1 dimensions}.
\newblock {\em Journal of Physics A Mathematical General}, 51(29):295202, 2018.

\bibitem{Manton1}
N.~S. {Manton}.
\newblock {Vortex solutions of the Popov equations}.
\newblock {\em Journal of Physics A Mathematical General}, 46(14):145402, April
  2013.

\bibitem{Baptista}
J.~M. Baptista.
\newblock {Vortices as degenerate metrics}.
\newblock {\em Lett. Math. Phys.}, 104:731--747, 2014.

\bibitem{MS}
N.~S. Manton and P.~M. Sutcliffe.
\newblock {\em Topological solitons.}
\newblock Cambridge monographs on mathematical physics. Cambridge University
  Press, Cambridge, July 2004.

\bibitem{JT}
A.~Jaffe and C.~Taubes.
\newblock {\em Vortices and monopoles: Structure of static gauge theories}.
\newblock Progress in physics. Birkh{\"{a}}user Boston, 1980.

\bibitem{Ross_thesis}
C.~ Ross.
\newblock {\em The Geometry of Integrable Vortices}.
\newblock PhD thesis, Heriot-Watt University, 2019.

\bibitem{Manton2}
N.~~S. Manton.
\newblock {Five Vortex Equations}.
\newblock {\em Journal of Physics A Mathematical General}, 50(12):125403, 2017.

\bibitem{CD}
F.~ Contatto and M.~ Dunajski.
\newblock {Manton’s five vortex equations from self-duality}.
\newblock {\em Journal of Physics A Mathematical General}, 50(37):375201, 2017.

\bibitem{Sharpe}
R.W. Sharpe.
\newblock {\em Differential Geometry: Cartan's Generalization of Klein's
  Erlangen Program}.
\newblock Graduate Texts in Mathematics. Springer New York, 2000.

\bibitem{Wise}
D.K Wise.
\newblock {\em Topological Gauge theory, Cartan geometry and gravity}.
\newblock PhD thesis, University of California (Riverside), 2007.

\bibitem{Witten1}
E.~ Witten.
\newblock Some exact multipseudoparticle solutions of classical yang-mills
  theory.
\newblock {\em Phys. Rev. Lett.}, 38:121--124, Jan 1977.

\bibitem{Popov}
A.~D. {Popov}.
\newblock {Integrable vortex-type equations on the two-sphere}.
\newblock {\em Phys. Rev. D.}, 86(10):105044, November 2012.

\bibitem{MR}
N.~S. {Manton} and N.~A. {Rink}.
\newblock {Vortices on hyperbolic surfaces}.
\newblock {\em Journal of Physics A Mathematical General}, 43:434024, October
  2010.

\bibitem{JS}
R.~{Jante} and B.~J. {Schroers}.
\newblock {Dirac operators on the Taub-NUT space, monopoles and SU(2)
  representations}.
\newblock {\em Journal of High Energy Physics}, 1:114, January 2014.

\bibitem{JS2}
R.~{Jante} and B.~J. {Schroers}.
\newblock {Spectral Properties of Schwarzschild Instantons}.
\newblock {\em Class. Quant. Grav.}, 33(20):205008, 2016.

\bibitem{BM}
J.C. Baez and J.P. Muniain.
\newblock {\em Gauge Fields, Knots, and Gravity}.
\newblock K \& E series on knots and everything. World Scientific, 1994.

\bibitem{HY}
P.~A. Horvathy and J.~C. Yera.
\newblock {Vortex solutions of the Liouville equation}.
\newblock {\em Lett. Math. Phys.}, 46:111--120, 1998.

\bibitem{Horvathy1}
P.~A. Horvathy.
\newblock {Lectures on (abelian) Chern-Simons vortices }.
\newblock [arXiv:0704.3220 [hep-th]].

\bibitem{HZ}
P.~A. Horvathy and P.~ Zhang.
\newblock {Vortices in (abelian) Chern-Simons gauge theory}.
\newblock {\em Phys. Rept.}, 481:83--142, 2009.

\bibitem{NW}
C.~R. Nappi and E.~ Witten.
\newblock {A WZW model based on a nonsemisimple group}.
\newblock {\em Phys. Rev. Lett.}, 71:3751--3753, 1993.

\bibitem{FM}
P.~Forgacs and N.~S. Manton.
\newblock {Space-Time Symmetries in Gauge Theories}.
\newblock {\em Commun. Math. Phys.}, 72:15, 1980.

\bibitem{walton_2021_exotic}
E.~ Walton. 
\newblock Exotic vortices and twisted holomorphic maps, 2021.
\newblock {[arXiv:2108.00315 [math-ph]]}

\bibitem{GR21}
S.~B. Gudnason and C.~ Ross.
\newblock {Magnetic Impurities, Integrable Vortices and the Toda Equation}.
\newblock {\em Lett. Math. Phys.}, 111:100, 2021.

\end{thebibliography}

\end{document}